\documentclass[submission,copyright,creativecommons,sharealike,noncommercial]{eptcs}
\providecommand{\event}{ITRS 2018}

\usepackage{color}

\usepackage{amsmath}
\usepackage{amssymb}
\usepackage{amsthm}
\usepackage{microtype}
\usepackage{bussproofs}
\usepackage{upgreek}
\usepackage{hyperref}

\title{Natural Deduction and Normalization Proofs for the Intersection Type Discipline}
\author{Federico Aschieri\thanks{Funded by FWF START project Y544--N23 and FWF P 32080--N31}\\
Institut f\"ur Logic and Computation\\ Technische Universit\"at Wien}
\def\titlerunning{Natural Deduction for Intersection Types} 
\def\authorrunning{F. Aschieri}

\providecommand{\urlalt}[2]{\href{#1}{#2}}
\providecommand{\doi}[1]{doi:\urlalt{http://dx.doi.org/#1}{#1}}
\providecommand{\urlalt}[2]{\href{#1}{#2}}

\newcommand{\comment}[1]{}

\newcommand{\DO}                       { {\mathsf{D\Omega}} }
\newcommand{\D}                       { {\mathsf{D}} }
\newtheorem{theorem}{Theorem}
\newtheorem{lemma}[theorem]{Lemma}
\newtheorem{definition}{Definition}
\newtheorem{proposition}{Proposition}

\begin{document}
\maketitle

\begin{abstract}
Refining and extending previous work by Retor\'e \cite{Retore}, we develop a systematic approach to intersection types via natural deduction. We show how a step of beta reduction can be seen as performing, at the level of typing derivations, Prawitz reductions in parallel. Then we derive as immediate consequences of Subject Reduction the main theorems about normalization for intersection types: for system $\D$, strong normalization, for system $\DO$, the leftmost reduction termination for terms typable without $\Omega$.
\end{abstract}

\section{Introduction}

One of the most remarkable properties of the intersection type systems $\DO$ and $\D$ \cite{CD} is that they characterize the normalizable and strongly normalizable terms of  $\lambda$-calculus. In turn, this characterization allows to prove in a logically grounded and elegant way several fundamental theorems about $\lambda$-calculus, like the \emph{uniqueness of normal forms} and the \emph{termination of the leftmost redex reduction} for normalizable terms \cite{Krivine}.  Unfortunately, since they exploit normalization for $\DO$, the first intersection-types-based proofs of these results employed the Tait reducibility technique \cite{Krivine}, very well known for not conveying any combinatorial  information and for  its logical complexity. Using reducibility to prove elementary theorems about $\lambda$-calculus is definitely an overkill and the resulting proofs are so indirect that are barely comprehensible. For these reasons, we are interested here in giving an elementary, direct, conceptually elegant proof of normalization for $\DO$, which may be used in elementary introductions to typed and untyped $\lambda$-calculus. With our approach, it will  turn out, the most important elementary results of $\lambda$-calculus can be proved as corollaries of the normalization theorem for simply typed $\lambda$-calculus.

The first elementary, arithmetical proof of strong normalization for system $\D$ was provided by Retor\'e \cite{Retore}; then several others followed (e.g. \cite{Bucciarelli, Valentini, David, Aschieri}). The beauty of Retor\'e's approach is that one sees that actually\dots there is nothing to prove. If one moves to a Prawitz-style natural deduction presentation of intersection types, instead of sticking to Gentzen-style natural deduction, as it is traditionally done \cite{Krivine}, then strong normalization becomes just consequence of Subject Reduction and normalization for natural deduction. The reason is that Gentzen-style natural deduction is based on sequents and, as a typing system, uses explicit contexts. As a result, the proof reductions are quite cumbersome to write and nowhere near the elegance achievable using Prawitz natural deduction trees. In fact, the usual proofs of Subject Reduction \cite{Krivine} provide no direct transformation of the typing derivation of a term into the typing derivations of the reducts. Writing the transformation explicitly, indeed, would be ugly. But without doing that, one misses the logical perspective on intersection types.

Retor\'e treated directly only strict intersection types, that is, types not allowing conjunction on the right of implications. This limitation is not present in  \cite{RonchiPimentel}, 
but the approach is not as direct and simple as Retor\'e's. Indeed,  strong normalization for $\D$ is deduced from strong normalization of a more complex logical system, where judgements are sequences of sequents, like in hypersequent calculi.  Since each proof of this system can be translated as a sequence of parallel natural deductions, its strong normalization can be derived.

 The goal of this paper is to address these limitations and extend the natural deduction approach to the typing system $\DO$, while refining Retor\'e's treatment of system $\D$ and removing the restriction to strict types. As corollaries, we shall obtain strong normalization for $\D$ and normalization for $\DO$.  A similar program was outlined in \cite{Dezani}, but never carried out in detail. Moreover, the  suggested normalization argument for $\DO$ is presented as standard cut-elimination technique, where one reduces at each step natural-deduction redexes of maximal complexity. This argument, however, is not suitable for proving termination of the leftmost reduction strategy for $\lambda$-calculus.  Achieving that requires subtle adjustments and carefully formulated inductive statements, and indeed none of the aforementioned natural-deduction-based works straightforwardly generalizes to $\DO$ so that it can accomplish our goals.

\section{Natural Deduction for Intersection Types}
In this section we define a natural deduction presentation of Coppo-Dezani intersection type systems $\D$ and $\DO$. These systems were invented with the aim of providing a logical characterization of strongly normalizable and normalizable $\lambda$-terms. Namely, the terms typable in $\D$ are exactly the strongly normalizable ones, while the terms typable in $\DO$ are precisely the normalizable ones.

We shall start by presenting the type inference rules, then we define a reduction relation on typing derivations, which consists in applying Prawitz reductions in parallel to the corresponding natural deduction.

\subsection{Typing Derivations}

A \textbf{typing tree} is a tree whose nodes are expressions of the form $t: A$, where $t$ is a $\lambda$-term and $A$ is a type built from type variables and $\top$, using the connectives $\rightarrow, \land$. A typing tree $\mathcal{D}$ with root $t: A$ will be denoted as  
$\AxiomC{$\mathcal{D}$}
\noLine
\UnaryInfC{$t: A$}
\DisplayProof$. With $\AxiomC{$x: A$}
\noLine
\UnaryInfC{$\mathcal{D}$}
\DisplayProof$, we denote a typing tree such that for  all leaves $x: A$ and $x: B$, we have $A=B$. 

A \textbf{typing  derivation} in system $\DO$ is a  typing tree obtained by means of the following inference rules.

$$
\AxiomC{}
\UnaryInfC{$x: A$}
\DisplayProof
\qquad
\AxiomC{}
\UnaryInfC{$t: \top$}
\DisplayProof
$$
$$
\AxiomC{$x: A$}
\noLine
\UnaryInfC{$\mathcal{D}$}
\noLine
\UnaryInfC{$u:B$}
\UnaryInfC{$\lambda x\, u: A\rightarrow B$}
\DisplayProof
\qquad
\AxiomC{$t: A\rightarrow B$}
\AxiomC{$u: A$}
\BinaryInfC{$t\, u: B$}
\DisplayProof
$$

\begin{prooftree}
\AxiomC{$t: A$}
\AxiomC{$t: B$}
\BinaryInfC{$t: A\land B$}
\end{prooftree}
$$
\AxiomC{$t: A\land B$}
\UnaryInfC{$t: A$}
\DisplayProof
\qquad
\AxiomC{$t: A\land B$}
\UnaryInfC{$t: B$}
\DisplayProof
$$
We observe that the only inference rule that requires a condition to be applied is the $\rightarrow$-introduction rule. One can only conclude $\lambda x\, u: A\rightarrow B$ when there is a typing derivation $\mathcal{D}$ of $u: B$ such that all the occurrences of $x$ are declared of the same type $A$. In general we allow variables to be declared of multiple types in the leaves of typing derivations, as polymorphism is the essence of intersection types. There is no special technical reason: we could very well have restricted variables to have unique types, but since there is no gain in doing this, we avoid the restriction.

By construction, if $x_{1}: A_{1}, \dots, x_{n}: A_{n}$, are all the leaves of  a typing derivation $\mathcal{D}$ of $t: A$ such that $x_{1}, \dots, x_{n}$ are free variables of $t$, then $\mathcal{D}$ is isomorphic to a natural deduction of $A$ from assumptions    $ A_{1}, \dots, A_{n}$; such a   $\mathcal{D}$ will be called \textbf{a typing derivation of $t: A$ from $x_{1}: A_{1}, \dots, x_{n}: A_{n}$}.    Not all natural deductions, of course, are isomorphic to  typing derivations: the $\land$-introduction can only be applied when the typing derivations of the premises type the \emph{same} term. 

As usual, $\D$ is the typing  system obtained from $\DO$ by dropping the rule $\AxiomC{}
\UnaryInfC{$t: \top$}
\DisplayProof
$. Moreover, we assume Barendregt's convention: in any context, free variables are always different from bound variables, so that there is no risk of capturing free variables when substituting terms for variables, as in $\beta$-reduction.

\subsection{Reduction Relation on Typing Derivations}

We now define the standard operation of derivation composition. With $\mathcal{D}[t/x]$ we shall denote the replacement of every occurrence of $x$ in $\mathcal{D}$ with the $\lambda$-term $t$. If $$
\AxiomC{$x: A$}
\noLine
\UnaryInfC{$\mathcal{D}$}
\noLine
\UnaryInfC{$u:B$}
\DisplayProof
\qquad
\AxiomC{$\mathcal{E}$}
\noLine
\UnaryInfC{$t: A$}
\DisplayProof$$
are two typing trees, the tree

$$\AxiomC{$\mathcal{E}$}
\noLine
\UnaryInfC{$t: A$}
\noLine
\UnaryInfC{$\mathcal{D}[t/x]$}
\noLine
\UnaryInfC{$u[t/x]:B$}
\DisplayProof
$$
denotes the typing tree obtained from $\mathcal{D}[t/x]$ by replacing each leaf $\AxiomC{}\UnaryInfC{$t: A$}\DisplayProof$ with the typing tree $\AxiomC{$\mathcal{E}$}
\noLine
\UnaryInfC{$t: A$}\DisplayProof$.
Indeed, this operation is a correct composition of typing derivations.

\begin{proposition}[Typing Derivation Composition]
Let $$
\AxiomC{$x: A$}
\noLine
\UnaryInfC{$\mathcal{D}$}
\noLine
\UnaryInfC{$u:B$}
\DisplayProof
\qquad
\AxiomC{$\mathcal{E}$}
\noLine
\UnaryInfC{$t: A$}
\DisplayProof$$
be two typing derivations. Then 
$$\AxiomC{$\mathcal{E}$}
\noLine
\UnaryInfC{$t: A$}
\noLine
\UnaryInfC{$\mathcal{D}[t/x]$}
\noLine
\UnaryInfC{$u[t/x]:B$}
\DisplayProof
$$
is a typing derivation
\end{proposition}
\begin{proof}
By straightforward induction on $\mathcal{D}$.
\end{proof}

In Table \ref{tab:derred}, we axiomatically define a binary reduction relation $\rightsquigarrow$ on typing derivations.  The lefthand-side derivations of the first two reductions are called respectively $\rightarrow$-redexes and $\land$-redexes. A typing derivation is $\land$-\textbf{normal} if it does not contain $\land$-redexes and it is \textbf{normal} if it is not in relation $\rightsquigarrow$ with any typing derivation.

 From the logical point of view, the relation $\rightsquigarrow$ formalizes the operation of performing several Prawitz reduction steps in parallel on the natural deduction associated to the typing derivation. Namely, we interpret the $\land$-introduction rule as a parallel composition of derivations. We remark that for the $\rightarrow$-elimination rule we do not allow parallel reductions, but this is a minimalistic design choice, rather than a necessity. From the computational point of view, indeed, the relation $\rightsquigarrow$ is intended to formalize the exact amount of Prawitz reductions needed to type a single step of $\beta$-reduction.
 
In the following we shall need the well-known notion of weak head reduction over $\lambda$-terms.

\begin{definition}[Weak Head Reduction] For every $\lambda$-term $t$, we say that $t\mapsto t'$ by \textbf{weak head reduction} if 
$$t= (\lambda x\, u)\, v\, t_{1}\dots t_{n}$$
$$t'= u[v/x]\, t_{1}\dots t_{n}$$

\end{definition}

\begin{table}[h]
\hrule 
$$\AxiomC{$x: A$}
\noLine
\UnaryInfC{$\mathcal{D}$}
\noLine
\UnaryInfC{$u:B$}
\UnaryInfC{$\lambda x\, u: A\rightarrow B$}
\AxiomC{$\mathcal{E}$}
\noLine
\UnaryInfC{$t: A$}
\BinaryInfC{$(\lambda x\, u) t: B$}
\DisplayProof
\qquad
\rightsquigarrow
\qquad
\AxiomC{$\mathcal{E}$}
\noLine
\UnaryInfC{$t: A$}
\noLine
\UnaryInfC{$\mathcal{D}[t/x]$}
\noLine
\UnaryInfC{$u[t/x]:B$}
\DisplayProof
$$

$$
\AxiomC{$\mathcal{D}_{1}$}
\noLine
\UnaryInfC{$t: A_{1}$}
\AxiomC{$\mathcal{D}_{2}$}
\noLine
\UnaryInfC{$t: A_{2}$}
\BinaryInfC{$t: A_{1}\land A_{2}$}
\UnaryInfC{$t: A_{i}$}
\DisplayProof
\qquad
\rightsquigarrow
\qquad
\AxiomC{$\mathcal{D}_{i}$}
\noLine
\UnaryInfC{$t: A_{i}$}
\DisplayProof
$$
$$\AxiomC{$\mathcal{D}$}
\noLine
\UnaryInfC{$t: B$}
\DisplayProof
\rightsquigarrow
\AxiomC{$\mathcal{D'}$}
\noLine
\UnaryInfC{$t': B$}
\DisplayProof
\qquad
\Rightarrow 
\qquad 
\AxiomC{$\mathcal{D}$}
\noLine
\UnaryInfC{$t: B$}
\UnaryInfC{$\lambda x\, t: A\rightarrow B$}
\DisplayProof
\rightsquigarrow
\AxiomC{$\mathcal{D}'$}
\noLine
\UnaryInfC{$t': B$}
\UnaryInfC{$\lambda x\, t': A\rightarrow B$}
\DisplayProof
$$
$$\AxiomC{$\mathcal{D}$}
\noLine
\UnaryInfC{$t: A_{1}\land A_{2}$}
\DisplayProof
\rightsquigarrow
\AxiomC{$\mathcal{D'}$}
\noLine
\UnaryInfC{$t': A_{1}\land A_{2}$}
\DisplayProof
\qquad
\Rightarrow 
\qquad 
\AxiomC{$\mathcal{D}$}
\noLine
\UnaryInfC{$t: A_{1}\land A_{2}$}
\UnaryInfC{$t: A_{i}$}
\DisplayProof
\rightsquigarrow
\AxiomC{$\mathcal{D}'$}
\noLine
\UnaryInfC{$t': A_{1}\land A_{2}$}
\UnaryInfC{$ t': A_{i}$}
\DisplayProof
$$
$$\AxiomC{$\mathcal{D}$}
\noLine
\UnaryInfC{$t: A$}
\DisplayProof
\rightsquigarrow
\AxiomC{$\mathcal{D}'$}
\noLine
\UnaryInfC{$t': A$}
\DisplayProof,
\AxiomC{$\mathcal{E}$}
\noLine
\UnaryInfC{$t: B$}
\DisplayProof
\rightsquigarrow
\AxiomC{$\mathcal{E}'$}
\noLine
\UnaryInfC{$t': B$}
\DisplayProof
\qquad
\Rightarrow 
\qquad 
\AxiomC{$\mathcal{D}$}
\noLine
\UnaryInfC{$t: A$}
\AxiomC{$\mathcal{E}$}
\noLine
\UnaryInfC{$t: B$}
\BinaryInfC{$t: A\land B$}
\DisplayProof
\rightsquigarrow
\AxiomC{$\mathcal{D}'$}
\noLine
\UnaryInfC{$t': A$}
\AxiomC{$\mathcal{E}'$}
\noLine
\UnaryInfC{$t': B$}
\BinaryInfC{$t': A\land B$}
\DisplayProof
$$
$$\AxiomC{$\mathcal{D}$}
\noLine
\UnaryInfC{$t: A$}
\DisplayProof
\rightsquigarrow
\AxiomC{$\mathcal{D}'$}
\noLine
\UnaryInfC{$t': A$}
\DisplayProof
\qquad
\Rightarrow 
\qquad 
\AxiomC{$\mathcal{D}$}
\noLine
\UnaryInfC{$t: A\rightarrow B$}
\AxiomC{$\mathcal{E}$}
\noLine
\UnaryInfC{$u: A$}
\BinaryInfC{$t\, u: B$}
\DisplayProof
\rightsquigarrow
\AxiomC{$\mathcal{D}'$}
\noLine
\UnaryInfC{$t': A\rightarrow B$}
\AxiomC{$\mathcal{E}$}
\noLine
\UnaryInfC{$u: A$}
\BinaryInfC{$t'\, u: B$}
\DisplayProof
$$
$$\AxiomC{$\mathcal{E}$}
\noLine
\UnaryInfC{$u: A$}
\DisplayProof
\rightsquigarrow
\AxiomC{$\mathcal{E}'$}
\noLine
\UnaryInfC{$u': A$}
\DisplayProof
\qquad
\Rightarrow 
\qquad 
\AxiomC{$\mathcal{D}$}
\noLine
\UnaryInfC{$t: A\rightarrow B$}
\AxiomC{$\mathcal{E}$}
\noLine
\UnaryInfC{$u: A$}
\BinaryInfC{$t\, u: B$}
\DisplayProof
\rightsquigarrow
\AxiomC{$\mathcal{D}$}
\noLine
\UnaryInfC{$t: A\rightarrow B$}
\AxiomC{$\mathcal{E}'$}
\noLine
\UnaryInfC{$u': A$}
\BinaryInfC{$t\, u': B$}
\DisplayProof
$$
\hrule
\caption{Reduction relation on typing derivations}\label{tab:derred}
\end{table}
\section{Subject Reduction}

The goal of this section is to prove Subject Reduction for system $\D$. We instead prove later a Subjection Reduction for system $\DO$, because in $\DO$ we are only interested in the contraction of the leftmost redex.

Since the relation $\rightsquigarrow$ embodies a Prawitz-style transformation of natural deductions, it always terminates.

\begin{proposition}
The reduction relation $\rightsquigarrow$ is strongly normalizing.
\end{proposition}
\begin{proof}
It is straightforward to prove, by induction on $\mathcal{D}$, that if $\mathcal{D}\rightsquigarrow\mathcal{D'}$, the natural deduction corresponding to  $\mathcal{D}$ reduces in a certain number of steps to the natural deduction corresponding to  $\mathcal{D}'$, using the standard Prawitz reductions for $\rightarrow$ and $\land$. Therefore,  the relation $\rightsquigarrow$  produces no infinite reduction path.
\end{proof}

Eliminating the useless $\land$-redex from typing derivations restores an important property of the simply typed $\lambda$-calculus: the only way to type a function with arrow type is by a $\rightarrow$-introduction.
\begin{proposition}[Introduce!]\label{prop:introduce}
Suppose $\mathcal{D}$ is a $\land$-normal typing derivation of $\lambda x\, u:  T\neq \top$ whose last rule is not an $\land$-introduction. Then 
$$\mathcal{D}=
\AxiomC{$x: A$}
\noLine
\UnaryInfC{$\mathcal{D}'$}
\noLine
\UnaryInfC{$u:B$}
\UnaryInfC{$\lambda x\, u: A\rightarrow B$}
\DisplayProof$$
with $T=A\rightarrow B$.


\end{proposition}

\begin{proof}
We proceed by induction on $\mathcal{D}$ and by cases according to the last rule of $\mathcal{D}$. We observe that the last rule cannot be a leaf nor an $\rightarrow$ elimination, because the conclusion of $\mathcal{D}$ is not a variable nor an application and $T\neq \top$. Therefore, only two rules can be applied:
\begin{itemize}
\item The last rule of $\mathcal{D}$ is a $\rightarrow$-introduction. This is the thesis.

\item 
$\mathcal{D}=
\AxiomC{$\mathcal{E}$}
\noLine
\UnaryInfC{$\lambda x\, u: B_{1}\land B_{2}$}
\UnaryInfC{$\lambda x\, u: B_{i}$}
\DisplayProof$, with $T=B_i$. We show that this case is impossibile. Since $\mathcal{D}$ is $\land$-normal, the last rule of $\mathcal{E}$ is not an $\land$-introduction. But by induction hypothesis the last rule of $\mathcal{E}$ must be an $\rightarrow$-introduction, which is a contradiction.
\end{itemize}
\end{proof}

We can now prove Subject Reduction in the usual way. The proof makes explicit the transformations that are implicit in the usual presentations of intersection types. 
\begin{theorem}[Subject Reduction for $\D$]\label{thm:subject}
Suppose $\mathcal{D}$ is a typing derivation of  $t: A$ in $\D$. Then
 $$t\mapsto t'
 \qquad \Rightarrow\qquad \AxiomC{$\mathcal{D}$}
\noLine
\UnaryInfC{$t: A$}
\DisplayProof
\rightsquigarrow^{+}
\AxiomC{$\mathcal{D}'$}
\noLine
\UnaryInfC{$t': A$}
\DisplayProof
$$
\begin{proof}
By straightforward induction on $\mathcal{D}$.
\end{proof}

\end{theorem}

\section{Strong Normalization}

As corollary of Subject Reduction we obtain strong normalization for system  $\D$.
\begin{theorem}[Strong Normalization]
Suppose that $\mathcal{D}$ is a  typing derivation  of $t: A$ in $\D$. Then $t$ is strongly normalizable.
\end{theorem}
\begin{proof}
We proceed by induction on the longest reduction of $\mathcal{D}$.  If $t\mapsto t'$, then by Theorem \ref{thm:subject} we obtain 
$
\AxiomC{$\mathcal{D}$}
\noLine
\UnaryInfC{$t: A$}
\DisplayProof
\rightsquigarrow^{}
\AxiomC{$\mathcal{D}'$}
\noLine
\UnaryInfC{$t': A$}
\DisplayProof
$ 
By induction hypothesis, $t$ is strongly normalizable. Since $t$ reduces only to strongly normalizable terms, it is strongly normalizable. 
\end{proof}

\section{Normalization by Leftmost Redex Reduction}
Before proving normalization for system $\DO$, we need a standard fact about intersection types.
\begin{proposition}\label{prop:subtype}
Suppose that $\mathcal{D}$ is an $\land$-normal typing derivation  of $x\, t_{1}\dots t_{n}: A\neq \top$ in $\DO$ from $x_{1}: A_{1}, \dots, x_{m}: A_{m}$ and the last rule of $\mathcal{D}$ is not an $\land$-introduction. Then $A$ is  a subformula of one among $A_{1}, \dots, A_{m}$.
\end{proposition}
\begin{proof}
By induction on $\mathcal{D}$ and by cases according to the last rule of $\mathcal{D}$. We must consider only the following cases.
\begin{itemize}
\item $n=0$ and $\mathcal{D}= 
\AxiomC{}
\UnaryInfC{$x: A_{i}$}
\DisplayProof$, with $x=x_{i}$ and  $A=A_{i}$. Then thesis is verified. 
\item $\mathcal{D}=
\AxiomC{$\mathcal{E}$}
\noLine
\UnaryInfC{$x\, t_{1}\dots t_{n-1}: B\rightarrow A$}
\AxiomC{$\mathcal{F}$}
\noLine
\UnaryInfC{$t_{n}: B$}
\BinaryInfC{$x\, t_{1}\dots t_{n-1}\, t_{n}: A$}\DisplayProof$. The last rule of $\mathcal{E}$ is not an $\land$-introduction, therefore by induction hypothesis $B\rightarrow A$ must be a  subformula of one among $A_{1}, \dots, A_{m}$, thus $A$ satisfies the thesis.
\item $\mathcal{D}=
\AxiomC{$\mathcal{E}$}
\noLine
\UnaryInfC{$x\, t_{1}\dots t_{n}: B_{1}\land B_{2}$}
\UnaryInfC{$x\, t_{1}\dots t_{n}: B_{i}$}
\DisplayProof$, with $A=B_i$. Since $\mathcal{D}$ is $\land$-normal, the last rule of $\mathcal{E}$ is not an $\land$-introduction, therefore by induction hypothesis $B_{1}\land B_{2}$ must be a  subformula of one among $A_{1}, \dots, A_{m}$, thus $B_{i}$ satisfies the thesis.

\end{itemize}

\end{proof}

We now prove the version of the  Subject Reduction that we  need for system $\DO$:  by contracting the leftmost redex of a typable $\lambda$-term $t$, we  induce a reduction of its typing derivation, provided $t$ is typable without $\top$. Intuitively, the leftmost redex cannot be inside a subterm of $t$ having type $\top$, the only case in which  we would not have any transformation of the natural deduction associated to the typing derivation.

\begin{lemma}[On the Left!]\label{lem:step}
Suppose that $\mathcal{D}$ is a $\land$-normal typing derivation  of $t: A\neq \top$ in $\DO$ from $x_{1}: A_{1}, \dots, x_{n}: A_{n}$. 
Then:
\begin{enumerate}
\item If the last rule of $\mathcal{D}$ is not a $\land$-introduction and   
 $t\mapsto t'$ by weak head reduction, then
$$
\AxiomC{$\mathcal{D}$}
\noLine
\UnaryInfC{$t: A$}
\DisplayProof
\rightsquigarrow^{+}
\AxiomC{$\mathcal{D}'$}
\noLine
\UnaryInfC{$t': A$}
\DisplayProof
$$

\item If $A_{1}, \ldots, A_{n}, A$ do not contain $\top$ and $t\mapsto t'$ by leftmost redex reduction, then 
$$
\AxiomC{$\mathcal{D}$}
\noLine
\UnaryInfC{$t: A$}
\DisplayProof
\rightsquigarrow^{+}
\AxiomC{$\mathcal{D}'$}
\noLine
\UnaryInfC{$t': A$}
\DisplayProof
$$

\end{enumerate}
\end{lemma}
\begin{proof}

We prove 1. and 2. simultaneously by induction on the size of $\mathcal{D}$ and by cases according to the last rule of $\mathcal{D}$. 
\begin{itemize}
\item $\mathcal{D}= \AxiomC{}
\UnaryInfC{$x: A$}
\DisplayProof$. This case is not possible, since $x$ does not reduce to any term.\\
\item $\mathcal{D}= \AxiomC{}
\UnaryInfC{$t: \top$}
\DisplayProof$, with $A=\top$. This case is not possible, since $A\neq\top$ by hypothesis.\\
\item 
$\mathcal{D}=
\AxiomC{$x: B$}
\noLine
\UnaryInfC{$\mathcal{D}'$}
\noLine
\UnaryInfC{$u:C$}
\UnaryInfC{$\lambda x\, u: B\rightarrow C$}
\DisplayProof$
with $t=\lambda x\, u$, $t'=\lambda x\, u'$, $u\mapsto u'$ and $A=B\rightarrow C$. We observe that 1. is trivially true, since $t$ reduces to nothing by weak head reduction.
Moving on to 2., we can apply the  induction hypothesis to 
$\AxiomC{$\mathcal{D}'$}
\noLine
\UnaryInfC{$u: C$}
\DisplayProof$, since $B$ does not contain $\top$ according to the hypotheses. We thus obtain 
$$
\AxiomC{$\mathcal{D}'$}
\noLine
\UnaryInfC{$u: C$}
\DisplayProof
\rightsquigarrow^{+}
\AxiomC{$\mathcal{D}''$}
\noLine
\UnaryInfC{$u': C$}
\DisplayProof
$$
Therefore, 
$$\mathcal{D}=
\AxiomC{$x: B$}
\noLine
\UnaryInfC{$\mathcal{D}'$}
\noLine
\UnaryInfC{$u: C$}
\UnaryInfC{$\lambda x\, u: B\rightarrow C$}
\DisplayProof
\rightsquigarrow^{+}
\AxiomC{$x: B$}
\noLine
\UnaryInfC{$\mathcal{D}''$}
\noLine
\UnaryInfC{$u': C$}
\UnaryInfC{$\lambda x\, u': B\rightarrow C$}
\DisplayProof$$
\item $\mathcal{D}=
\AxiomC{$\mathcal{E}$}
\noLine
\UnaryInfC{$u: B\rightarrow A$}
\AxiomC{$\mathcal{F}$}
\noLine
\UnaryInfC{$v: B$}
\BinaryInfC{$u\, v: A$}\DisplayProof$, with $t= u\, v$. 

If $t\mapsto t'$ by weak head reduction, either $u\mapsto u'$ by weak head reduction and $t'=u'\, v$, or $u=\lambda x\, w$ and $t'= w[v/x]$. In the first case,  the last rule of $\mathcal{E}$ cannot be an $\land$-introduction, so by induction hypothesis 1. we obtain
$$
\AxiomC{$\mathcal{E}$}
\noLine
\UnaryInfC{$u: B\rightarrow A$}
\DisplayProof
\rightsquigarrow^{+}
\AxiomC{$\mathcal{E}'$}
\noLine
\UnaryInfC{$u': B\rightarrow A$}
\DisplayProof
$$
therefore,
$$
\AxiomC{$\mathcal{E}$}
\noLine
\UnaryInfC{$u: B\rightarrow A$}
\AxiomC{$\mathcal{F}$}
\noLine
\UnaryInfC{$v: B$}
\BinaryInfC{$u\, v: A$}\DisplayProof
\rightsquigarrow^{+}
\AxiomC{$\mathcal{E}'$}
\noLine
\UnaryInfC{$u': B\rightarrow A$}
\AxiomC{$\mathcal{F}$}
\noLine
\UnaryInfC{$v: B$}
\BinaryInfC{$u'\, v: A$}\DisplayProof$$
which proves 1. and 2. In the second case, since $\mathcal{E}$ is by hypothesis $\land$-normal, by Proposition \ref{prop:introduce}, 
$$\mathcal{E}=
\AxiomC{$x: B$}
\noLine
\UnaryInfC{$\mathcal{E}'$}
\noLine
\UnaryInfC{$w:A$}
\UnaryInfC{$\lambda x\, w: B\rightarrow A$}
\DisplayProof
$$
Therefore,
$$
\AxiomC{$\mathcal{E}$}
\noLine
\UnaryInfC{$u: B\rightarrow A$}
\AxiomC{$\mathcal{F}$}
\noLine
\UnaryInfC{$v: B$}
\BinaryInfC{$u\, v: A$}\DisplayProof
=
\AxiomC{$x: B$}
\noLine
\UnaryInfC{$\mathcal{E}'$}
\noLine
\UnaryInfC{$w: A$}
\UnaryInfC{$\lambda x\, w: B\rightarrow A$}
\AxiomC{$\mathcal{F}$}
\noLine
\UnaryInfC{$v: B$}
\BinaryInfC{$(\lambda x\, w)\, v: A$}\DisplayProof
\rightsquigarrow^{+}
\AxiomC{$\mathcal{F}$}
\noLine
\UnaryInfC{$v: B$}
\noLine
\UnaryInfC{$\mathcal{E}'[v/x]$}
\noLine
\UnaryInfC{$w[v/x]:A$}
\DisplayProof
$$
which proves 1. and 2. 

We can now assume that $t\mapsto t'$ not by weak head reduction, thus we are left to prove 2. Since, $t$ has no head redex, $t=x\, t_{1}\dots t_{m}$, with $u=x\, t_{1}\dots t_{m-1}$ and $v=t_{m}$.  By Proposition \ref{prop:subtype}, $B\rightarrow A$ must be a  subformula  of some $A_{i}$. Therefore $B\rightarrow A$ and $B$ do not contain $\top$. Since $t'= u'\, v$, with $u\mapsto u'$, or $t'= u\, v'$, with $v\mapsto v'$, by induction hypothesis  respectively 
$$
\AxiomC{$\mathcal{E}$}
\noLine
\UnaryInfC{$u: B\rightarrow A$}
\DisplayProof
\rightsquigarrow^{+}
\AxiomC{$\mathcal{E}'$}
\noLine
\UnaryInfC{$u': B\rightarrow A$}
\DisplayProof
$$
or 
$$
\AxiomC{$\mathcal{F}$}
\noLine
\UnaryInfC{$v: B$}
\DisplayProof
\rightsquigarrow^{+}
\AxiomC{$\mathcal{F}'$}
\noLine
\UnaryInfC{$v': B$}
\DisplayProof
$$
Therefore,
$$
\AxiomC{$\mathcal{E}$}
\noLine
\UnaryInfC{$u: B\rightarrow A$}
\AxiomC{$\mathcal{F}$}
\noLine
\UnaryInfC{$v: B$}
\BinaryInfC{$u\, v: A$}\DisplayProof
\rightsquigarrow^{+}
\AxiomC{$\mathcal{E}'$}
\noLine
\UnaryInfC{$u': B\rightarrow A$}
\AxiomC{$\mathcal{F}$}
\noLine
\UnaryInfC{$v: B$}
\BinaryInfC{$u'\, v: A$}\DisplayProof$$
or 
$$
\AxiomC{$\mathcal{E}$}
\noLine
\UnaryInfC{$u: B\rightarrow A$}
\AxiomC{$\mathcal{F}$}
\noLine
\UnaryInfC{$v: B$}
\BinaryInfC{$u\, v: A$}\DisplayProof
\rightsquigarrow^{+}
\AxiomC{$\mathcal{E}$}
\noLine
\UnaryInfC{$u: B\rightarrow A$}
\AxiomC{$\mathcal{F}'$}
\noLine
\UnaryInfC{$v': B$}
\BinaryInfC{$u\, v': A$}\DisplayProof$$
which is the thesis.

\item $\mathcal{D}=
\AxiomC{$\mathcal{E}$}
\noLine
\UnaryInfC{$t: C$}
\AxiomC{$\mathcal{F}$}
\noLine
\UnaryInfC{$t: B$}
\BinaryInfC{$t: C\land B$}
\DisplayProof$, with $A=C\land B$.  By hypothesis on $\mathcal{D}$, 1. is trivially true, we thus prove 2. Since $C\land B$ does not contain $\top$, also $C$ and $B$ do not contain $\top$. Therefore by induction hypothesis 2., we get
$$
\AxiomC{$\mathcal{E}$}
\noLine
\UnaryInfC{$t: C$}
\DisplayProof
\rightsquigarrow^{+}
\AxiomC{$\mathcal{E}'$}
\noLine
\UnaryInfC{$t': C$}
\DisplayProof
$$ 
$$
\AxiomC{$\mathcal{F}$}
\noLine
\UnaryInfC{$t: B$}
\DisplayProof
\rightsquigarrow^{+}
\AxiomC{$\mathcal{F}'$}
\noLine
\UnaryInfC{$t': B$}
\DisplayProof
$$
thus
$$\AxiomC{$\mathcal{E}$}
\noLine
\UnaryInfC{$t: C$}
\AxiomC{$\mathcal{F}$}
\noLine
\UnaryInfC{$t: B$}
\BinaryInfC{$t: C\land B$}
\DisplayProof
\rightsquigarrow^{+}
\AxiomC{$\mathcal{E}'$}
\noLine
\UnaryInfC{$t': C$}
\AxiomC{$\mathcal{F}'$}
\noLine
\UnaryInfC{$t': B$}
\BinaryInfC{$t': C\land B$}
\DisplayProof$$
which proves 2.

\item $\mathcal{D}=
\AxiomC{$\mathcal{E}$}
\noLine
\UnaryInfC{$t: B_{1}\land B_{2}$}
\UnaryInfC{$t: B_{i}$}
\DisplayProof$, with $A=B_{i}$. We first observe that $t$ cannot start with $\lambda$, otherwise, since $\mathcal{D}$ is $\land$-normal and thus the last rule of $\mathcal{E}$ cannot be an $\land$-introduction,  by  Proposition \ref{prop:introduce} we would obtain that the last rule of $\mathcal{E}$ is a $\rightarrow$-introduction.

Now, if $t\mapsto t'$ by weak head reduction, then by induction hypothesis 1., we get 
$$
\AxiomC{$\mathcal{E}$}
\noLine
\UnaryInfC{$t: B_{1}\land B_{2}$}
\DisplayProof
\rightsquigarrow^{+}
\AxiomC{$\mathcal{E}'$}
\noLine
\UnaryInfC{$t': B_{1}\land B_{2}$}
\DisplayProof
$$ 
thus
$$\AxiomC{$\mathcal{E}$}
\noLine
\UnaryInfC{$t: B_{1}\land B_{2}$}
\UnaryInfC{$t: B_{i}$}
\DisplayProof
\rightsquigarrow^{+}
\AxiomC{$\mathcal{E}'$}
\noLine
\UnaryInfC{$t': B_{1}\land B_{2}$}
\UnaryInfC{$t': B_{i}$}
\DisplayProof
$$
which proves 1. and 2.
Therefore, we can assume that $t\mapsto t'$ not by weak head reduction and we are left to prove 2. Since $t$ does not start with $\lambda$, it  has no head redex, thus $t= x\, t_{1}\dots t_{m}$. By Proposition \ref{prop:subtype}, applied to $\mathcal{E}$, we obtain that $B_{1}\land B_{2}$ is a subformula of some among $A_{1}, \dots, A_{n}$, hence $B_{1}\land B_{2}$ cannot contain $\top$. By induction hypothesis, 
$$
\AxiomC{$\mathcal{E}$}
\noLine
\UnaryInfC{$t: B_{1}\land B_{2}$}
\DisplayProof
\rightsquigarrow^{+}
\AxiomC{$\mathcal{E}'$}
\noLine
\UnaryInfC{$t': B_{1}\land B_{2}$}
\DisplayProof
$$ 
and we obtain the thesis.
\end{itemize}
\end{proof}
We now prove that every term typable in $\DO$ without $\top$ is normalizable by leftmost redex reduction. The natural deduction proof sheds new light on this fundamental result. Every reduction step contracting the leftmost redex is actually a combination of reduction steps at the level of the natural deduction corresponding to the typing derivation. When this natural deduction reaches normal form, the term is in normal form. We also remark that the subformula property must hold.  Since the term is typable without $\top$, the normal derivation is actually a derivation in system $\D$! This means that  the subterms having type $\top$ are systematically erased.

\begin{theorem}[Normalization by Leftmost Redex Reduction]\label{thm:norm}
Suppose that $\mathcal{D}$ is a  typing derivation  of $t: A$ in $\DO$ from $x_{1}: A_{1}, \dots, x_{n}: A_{n}$ 
such that $A_{1}, \ldots, A_{n}, A$ do not contain $\top$. Then the leftmost redex reduction of $t$ terminates.
\end{theorem}
\begin{proof}
We proceed by induction on the longest reduction of $\mathcal{D}$. We have that $
\AxiomC{$\mathcal{D}$}
\noLine
\UnaryInfC{$t: A$}
\DisplayProof
\rightsquigarrow^{*}
\AxiomC{$\mathcal{D}'$}
\noLine
\UnaryInfC{$t: A$}
\DisplayProof
$, where $\mathcal{D}'$ is $\land$-normal. If $t$ is normal we are done. If $t\mapsto t'$ by leftmost redex reduction, by Lemma \ref{lem:step} we obtain 
$
\AxiomC{$\mathcal{D}'$}
\noLine
\UnaryInfC{$t: A$}
\DisplayProof
\rightsquigarrow^{+}
\AxiomC{$\mathcal{D}''$}
\noLine
\UnaryInfC{$t': A$}
\DisplayProof
$ 
By induction hypothesis, the leftmost redex reduction of $t$ terminates, which yields the thesis.
\end{proof}


\begin{thebibliography}{[1]}

\bibitem{Aschieri} F. Aschieri, \emph{Una caratterizzazione dei lambda termini non fortemente normalizzabili}, Master Thesis, Universit\`a di Verona, 2007.

\bibitem{Bucciarelli} {A. Bucciarelli, S. De Lorenzis,
               A. Piperno, I. Salvo},
\emph{Some Computational Properties of Intersection Types},
14th Annual {IEEE} Symposium on Logic in Computer Science, Trento,
  pp.  {109--118},  {1999}. \doi{10.1109/LICS.1999.782598}.

\bibitem{CD} M. Coppo, M. Dezani,
\emph{A new type assignment for {\(\lambda\)}-terms}, Arch. Math. Log.,
  vol. {19}, n. {1}, pp. {139--156}, 1978. \doi{10.1007/BF02011875}
  
\bibitem{David}  R. David, \emph{Normalization without reducibility},
{Ann. Pure Appl. Logic},
  vol.  {107}, n.  {1-3},  pp.  {121--130}, {2001}. \doi{10.1016/S0168-0072(00)00030-0}.

\bibitem{Dezani} M. Dezani-Ciancaglini, E. Giovannetti, U. de'Liguoro, \emph{Intersection Types, $\lambda$-models, and B\"ohm Trees}. Theories of Types and Proofs, 45--97, The Mathematical Society of Japan, Tokyo, Japan, 1998. \doi{10.2969/msjmemoirs/00201C020}. 

\bibitem{Krivine}
J.-L. Krivine, \emph{Lambda-calculus, types and models}, Studies in Logic and Foundations of Mathematics, pp. 1--176,  Masson, Paris, 1990.

\bibitem{RonchiPimentel} E. Pimentel,
               S. Ronchi Della Rocca and
               L. Roversi,
 \emph{Intersection Types from a Proof-theoretic Perspective},
Fundamenta Informaticae, vol. {121}, n. 1--4,  pp. {253--274}, {2012}. \doi{10.3233/FI-2012-778}.

\bibitem{Retore}
C. Retor\'e, \emph{A Note on Intersection Types}, Rapport de Recherche, RR-2431, INRIA, 1994. 

\bibitem{vanBakel} S. van Bakel,
\emph{Complete Restrictions of the Intersection Type Discipline},
 {Theor. Comput. Sci.}, v. {102}, n.  {1}, pp. {135--163}, 1992. \doi{10.1016/0304-3975(92)90297-S}
 
\bibitem{Valentini} S. Valentini, \emph{An elementary proof of strong normalization for intersection types}, {Arch. Math. Log.},
  vol, {40}, n.  {7}, pp. {475--488}, {2001}. \doi{10.1007/s001530000070}

\end{thebibliography}
\end{document}